\RequirePackage{fix-cm}
\documentclass[smallextended]{svjour3}       
\smartqed  
\usepackage{graphicx}
\usepackage{fullpage,graphicx}
\usepackage{stmaryrd}
\usepackage{amsmath}
\usepackage{amssymb}
\usepackage{dsfont}
\usepackage{color}
\usepackage{subfigure}
\usepackage[ruled,vlined]{algorithm2e}

\begin{document}

\title{The dynamics of the CBC Mode of Operation}
\author{Abdessalem Abidi$^1$, Christophe Guyeux$^2$, Bechara AL Bouna$^3$, Belgacem Bouall\`gue$^1$, and Mohsen Machhout$^1$\\
$^1$~Electronics and Microelectronics Laboratory,
Faculty of Sciences of Monastir, University of Monastir, Tunisia
\\ $^2$~FEMTO-ST Institute, UMR 6174 CNRS, DISC Computer Science Department \\
    University of Franche-Comt\'e,
    16, Route de Gray, 25000 Besan\c{c}on, France\\
$^3$~TICKET Lab., Antonine University, Hadat-Baabda, Lebanon\\
 Corresponding author:    abdessalemabidi9@gmail.com 
}

\date{Received: date / Accepted: date}

\maketitle

\begin{abstract}
In cryptography, the Cipher Block Chaining (CBC), one of the most commonly used mode in recent years, is a mode of operation that uses a block cipher to provide confidentiality or authenticity. In our previous research work, we have shown that this mode of operation exhibits,  under some conditions, a chaotic behaviour. We have studied this behaviour by evaluating both its level of sensibility and expansivity. In this paper, we intend to deepen the topological study of the CBC mode of operation and evaluate its property of topological mixing. Additionally, other quantitative evaluations are performed, and the level of topological entropy has been evaluated too.

\end{abstract}

\section{Introduction}
\label{intro}

Block ciphers, like the Data Encryption Standard (DES) or the Advanced Encryption Standard (AES), have a very simple principle. They do not treat the original text bit by bit but they manipulate blocks of text. In other words, the original text is broken into blocks of n bits. For each block, the encryption algorithm is applied to obtain an encrypted block that has the same size. After that, we put together all of these blocks, which are encrypted separately, to obtain the complete encrypted message. For decryption, we proceed in the same way, but this time, we start from the encrypted text to obtain the original one using now the decryption algorithm instead of the encryption function. So, it is not sufficient to put anyhow a block cipher algorithm in a program. We can, instead, use these algorithms in various ways according to their specific needs. These ways are called the block cipher modes of operation. There are several modes of operation and each one possesses its own characteristics and its specific security properties. In this article, we are only interested in one of these modes, namely the cipher block chaining (CBC) mode, and we will study its dynamical behavior of chaos.

The chaos theory that we consider in this article is the Devaney's topological one and its ramifications~\cite{devaney}. Being reputed as one of the best mathematical definition of chaos, this theory offers a framework with qualitative and quantitative tools to evaluate the notion of unpredictability~\cite{bahi2011efficient}. As an application of our fundamental results, we are interested in the area of information safety and security. Specifically, our contribution belongs to the field of the cipher block chaining modes of operation.

In~\cite{Abdessalem2016}, we have started to give mathematical proofs that emphasize the chaotic behavior of the CBC mode of operation. Thereafter, in~\cite{abidi2016quantitative}, we have stated that in addition to being chaotic as defined in Devaney's formulation, this mode is indeed largely sensible to initial errors or modification on either the Initialization Vector IV or the message to encrypt. Its expansivity has been regarded too, but this property is not satisfied as it has been established thanks to a counter example. In this new article,  we intend to deepen the topological study of this CBC mode of operation in order to obtain a complete mathematical overview of its dynamics. 
  
The remainder of this research work is organized as follows. In the next section, we will recall some basic definitions related to chaos and cipher block chaining mode of operation. Previously obtained results are recalled in Section \ref{section:previous results}. Sections \ref{topo} and \ref{sec:entro} contain the main contribution of this article: the first one evaluates the topological mixing of the CBC mode of operation, while the second one focuses on its topological entropy. This article ends with a conclusion section where our contribution is summarized and intended future work is presented.

\section{Basic recalls}
\label{section:BASIC RECALLS}
This section is devoted to basic definitions and terminologies in the field of topological chaos and in the one of block cipher mode of operation.

\subsection{Devaney's Chaotic Dynamical Systems}
\label{subsec:Devaney}
In the remainder of this article,

$m_n$ denotes the $n^{th}$ block message of a sequence $S$ while $m^j$ stands for the $j-th$ bit of integer of the block message $m\in \llbracket 0, 2^{\mathbb{N}}-1 \rrbracket$, expressed in the binary numeral system and $x_{i}$ stands for the $i^{th}$ component of a vector $x$. 

$\mathbb{X}^\mathbb{N}$ is the set of all sequences whose elements belong to $\mathbb{X}$.

$f^{\circ k}=f\circ ...\circ f$ is for the $k^{th}$ composition of a function $f$.
$\mathbb{N}$ is the set of natural (non-negative) numbers, while $\mathbb{N}^*$ stands for the positive integers $1, 2, 3, \hdots$ 

Finally, the following
notation is used: $\llbracket1;N\rrbracket=\{1,2,\hdots,N\}$.

Consider a topological space $(\mathcal{X},\tau)$, where $\tau$ represents a family of subsets of $\mathcal{X}$, and a continuous function $f :
\mathbb{X} \rightarrow \mathcal{X}$ on $(\mathbb{X},\tau)$.

\begin{definition}
The function $f$ is \emph{topologically transitive} if, for any pair of nonempty open sets
$\mathcal{U},\mathcal{V} \subset \mathcal{X}$, there exists an integer $k>0$ such that $f^{\circ k}(\mathbb{U}) \cap \mathbb{V} \neq
\varnothing$.
\end{definition}

\begin{definition}
An element $x$ is a \emph{periodic point} for $f$ of period $n\in \mathbb{N}$, $n>1$,
if $f^{\circ n}(x)=x$ and $f^{\circ k}(x) \neq x, 1\le k\le n$. 
\end{definition}

\begin{definition}
$f$ is  \emph{regular} on $(\mathbb{X}, \tau)$ if the set of periodic
points for $f$ is dense in $\mathbb{X}$: for any point $x$ in $\mathbb{X}$,
any neighborhood of $x$ contains at least one periodic point.
\end{definition}

\begin{definition}
\label{sensitivity} The function $f$ has \emph{sensitive dependence on initial conditions} on the metric space $(\mathcal{X},d)$
if there exists $\delta >0$ such that, for any $x\in \mathcal{X}$ and any
neighborhood $\mathcal{V}$ of $x$, there exist $y\in \mathcal{V}$ and $n > 0$ such that
the distance $d$ between the results of their $n^{th}$ composition, $f^{\circ n}(x)$ and $f^{\circ n}(y)$, is greater than $\delta$:
$$d\left(f^{\circ n}(x), f^{\circ n}(y)\right) >\delta .$$
$\delta$ is called the \emph{constant of sensitivity} of $f$.
\end{definition}

\begin{definition}[Devaney's formulation of chaos~\cite{devaney}]
\label{def:dev}
The function $f$ is  \emph{chaotic} on a metric space $(\mathcal{X},d)$ if $f$ is regular,
topologically transitive, and has sensitive dependence on initial conditions.
\end{definition}

Banks \emph{et al.} have proven in~\cite{Banks92} that when $f$ is regular and transitive on a metric space $(\mathcal{X}, d)$, then $f$ has the property of sensitive dependence on initial conditions. This is why chaos can be formulated too in a topological space $(\mathcal{X}, \tau)$: in that situation, chaos is obtained when $f$ is regular and
topologically transitive.
Note that the transitivity property is often obtained as a consequence of the strong transitivity one, which is defined below.

\begin{definition}
\label{def:strongTrans}
$f$ is \emph{strongly transitive} on $(\mathcal{X},d)$ if, for all point $x,y \in \mathbb{X}$ and for all neighborhood $\mathbb{V}$ of $x$, it exists $n \in \mathbb{N}$ and $x'\in \mathbb{V}$ such that $f^{\circ n}(x')=y$. 
\end{definition}

 \subsection{CBC properties}
 \label{sec:CBC properties}
 
  As what has been already defined, a mode of operation is an algorithm that uses a block cipher to provide an information service such as confidentiality or authenticity. The most commonly used mode of operation in recent decades is the cipher block chaining CBC. In what follows, we will see how this mode works in practice.
  \subsubsection{Initialisation vector IV}
  Like some other modes of operation, the CBC one requires not only a plaintext but also an initialization vector (denoted as IV in what follows). 
  An IV is an arbitrary number that must be generated for each execution of the encryption operation. For the decryption algorithm, the vector used must be the same, see Figure 1. The use of this vector prevents repetition in data encryption. So, it offers the benefit to make this operation more difficult for a hacker. Indeed, it deprives him to find patterns and break a cipher. The length of the initialization vector (the number of bits or bytes it contains) depends on the method of encryption. It is usually comparable to the length of the encryption key or block of the cipher in use. 
  In general, the IV (or information sufficient to determine it) possesses the following characteristics:
  \begin{itemize}
\item The IV must be available to each party of communication.
\item It does not need to be secret. So, it may be transmitted with the cipher text.
\item This vector must be unpredictable: for any given plaintext, it must not be possible to predict the IV that will be associated to it.
\end{itemize}

  \subsubsection{CBC mode characteristics}
  In cryptography, Cipher Block Chaining is a block cipher mode that provides confidentiality but not message integrity. It offers a solution to most of the problems presented by the Electronic Code Book (ECB) mode; thanks to CBC mode, the encryption will depend not only on the plaintext, but also on all preceding blocks.
 More precisely, each block of plaintext is XORed immediately with the previous cipher text block before being encrypted (\textit{i.e.}, the binary operator XOR is applied between two stated blocks). For the first block, the initialization vector acts as the previous cipher text block, see Figure 1.
   CBC mode possesses several advantages. In fact, the same plaintext is encrypted differently in the case of different initialization vectors. In addition, the encryption operation of each block depends on the preceding one, so any modification in the order of the cipher text block makes the decryption operation unrealizable. Furthermore, if a transmission error affects the encrypted data, for example the $C_{i}$ block, so only $m_{i}$ and $m_{i+1}$ blocks will be influenced by this error while the  remained blocks will be determined correctly.
Conversely, the CBC mode is characterized by two main drawbacks. The first one is that encryption is sequential (\textit{i.e.}, it cannot be parallelized). The second one is that the initial message should be padded to a multiple of the cipher block size.

\begin{figure}[h]
    \centering
 \subfigure[CBC encryption mode]{\label{fig:CBCenc}
        \includegraphics[scale=0.5]{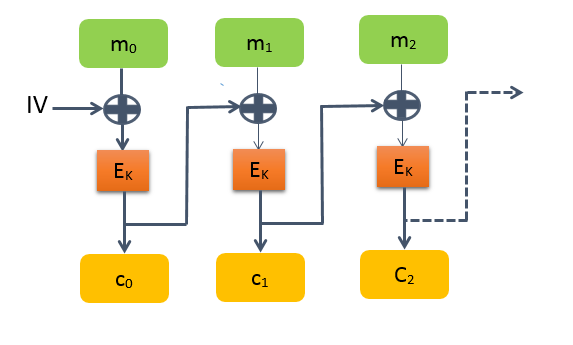}}     \subfigure[CBC decryption mode]{\includegraphics[scale=0.5]{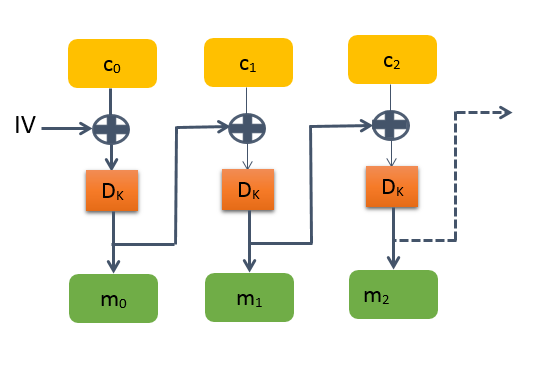}}
    \caption{CBC mode of operation}
     \label{fig:CBC}
\end{figure}
In the next section, we will summarize our previous results that have been detailed respectively in~\cite{Abdessalem2016} and~\cite {abidi2016quantitative}.

\section{Previously obtained results}
\label{section:previous results}
 
\subsection{Modeling the CBC mode as a dynamical system}
Our modeling follows a same canvas as what has been done for hash functions~\cite{bg10:ij,gb11:bc} or pseudorandom number generation~\cite{bfgw11:ij}.

Let us consider the CBC mode of operation with a keyed encryption function 
$\mathbb{E}_{\kappa}:\mathbb{B}^\mathbb{N} \rightarrow \mathbb{B}^\mathbb{N} $ 
depending on a secret key 
$\kappa$, where $\mathbb{N}$ is the size for the block cipher, and $\mathbb{D}_{\kappa}:\mathbb{B}^\mathbb{N} \rightarrow \mathbb{B}^\mathbb{N}$ 
is the associated  decryption function, which is such that 
$\forall \kappa, \mathbb{E}_{\kappa} \circ \mathbb{D}_{\kappa}$ 
is the identity function. We define  
the Cartesian product $\mathbb{X}=\mathbb{B}^\mathbb{N}\times\mathbb{S}_\mathbb{N}$, where:

\begin{itemize}
\item $\mathbb{B} = \{0,1\}$ is the set of Boolean values,
\item $\mathbb{S}_\mathbb{N} = \llbracket 0, 2^\mathbb{N}-1\rrbracket^\mathbb{N}$, the set of infinite sequences of natural integers bounded by $2^\mathbb{N}-1$, or the set of infinite $\mathbb{N}$-bits block messages,
\end{itemize}

in such a way that $\mathbb{X}$ is constituted by couples: the internal states of the mode of operation, and sequences of block messages.
Let us consider the initial function:
$$\begin{array}{cccc}
\iota:& \mathbb{S}_\mathbb{N} & \longrightarrow & \llbracket 0, 2^\mathbb{N}-1 \rrbracket \\
 & (m_i)_{i \in \mathbb{N}} & \longmapsto & m_0
\end{array}$$
that returns the first block of a (infinite) message, and the shift function:
$$\begin{array}{cccc}
 \sigma:& \mathbb{S}_\mathbb{N} & \longrightarrow & \mathcal{S}_\mathsf{N} \\
 & (m_0, m_1, m_2, ...) & \longmapsto & (m_1, m_2, m_3, ...)
\end{array}$$
that removes the first block of a message, when counting from the left. We define:
 
$$\begin{array}{cccc}
F_f:& \mathbb{B}^\mathbb{N}\times \llbracket 0, 2^\mathbb{N}-1 \rrbracket & \longrightarrow & \mathbb{B}^\mathbb{N}\\
 & (x,m) & \longmapsto & \left(x_j m^j + f(x)_j \overline{m^j}\right)_{j=1..\mathsf{N}} 
\end{array} .$$
This function returns the inputted binary vector $x$, whose $m^j$-th components $x_{m^j}$ have been replaced by $f(x)_{m^j}$, for all $j=1\dots \mathbb{N}$ such that $m^j=0$. In case where $f$ is the vectorial negation, this function will correspond to one XOR between the plaintext and the previous encrypted state.  
The CBC mode of operation can be rewritten as the following dynamical system:
\begin{equation}
\mathbb{\label{eq:sysdyn}
\left\{
\begin{array}{ll}
X^0 = & (IV,m)\\
X^{n+1} = & \left(\mathcal{E}_{\kappa} \circ F_{f_0} \left( \iota(X_1^n), X_2^n\right), \sigma (X_1^n)\right)
\end{array}
\right.}
\end{equation}

For any given $g:\llbracket 0, 2^\mathsf{N}-1 \rrbracket \times \mathds{B}^\mathsf{N} \longrightarrow \mathds{B}^\mathsf{N}$, 
we denote 
$G_g(X) = \left(g(\iota(X_1),X_2);\sigma (X_1)\right)$ (when $g = \mathbb{E}_{\kappa}\circ F_{f_0}$, 
we obtain one cipher block of the CBC, as depicted in Figure~\ref{fig:CBC}). The recurrent relation of Eq.\eqref{eq:sysdyn} can be rewritten in a condensed way, as follows.
\begin{equation}
X^{n+1} = G_{\mathcal{E}_{\kappa}\circ F_{f_0}} \left(X^n\right) .
\end{equation}
With such a rewriting, one iterate of the discrete dynamical system above corresponds exactly to one cipher block in the CBC mode of operation. Note that the second component of this system is a subshift of finite type that is related to the symbolic dynamical systems known for their relation with chaos~\cite{lind1995introduction}. We now define a distance on $\mathcal{X}$ as follows: $d((x,m);(\check{x},\check{m})) = d_e(x,\check{x})+d_m(m,\check{m})$, where:
$$\left\{\begin{array}{ll}
d_e(x,\check{x})  & = \sum_{k=1}^\mathsf{N} |x_k - \check{x}_k|  \\
&\\
d_m(m,\check{m}) & = \displaystyle{\dfrac{9}{\mathsf{N}} \sum_{k=1}^\infty \dfrac{\sum_{i=1}^\mathsf{N} \left|m^i - \check{m}^i\right|}{10^k}} .
\end{array}\right.$$
This distance has been introduced to satisfy the following requirements:
\begin{itemize}
\item The integral part between two points $X,Y$ of the 
topological space $\mathcal{X}$ corresponds to the number of binary components that are different between the two internal states $X_1$ and $Y_1$.
\item The $k$-th digit in the decimal part of the distance between $X$ and $Y$ is equal to 0 if and only if the $k$-th blocks of messages $X_2$ and $Y_2$ are equal. This desire is at the origin of the normalization factor $\dfrac{9}{\mathsf{N}}$.
\end{itemize}

\subsection{Proofs of chaos}
As mentioned in Definition~\ref{def:dev},  a function $f$ is  \emph{chaotic} on $(\mathcal{X},\tau)$ if $f$ is regular and
topologically transitive. 
We have began in~\cite{Abdessalem2016} by stating some propositions that are primarily required in order to proof the chaotic behavior of the CBC mode of operation.

\begin{proposition}
\label{prop:transitivity}
Let $g=\mathcal{E}_{\kappa} \circ F_{f_0}$, where $\mathcal{E}_{\kappa}$ 
is a given keyed block cipher and 
$f_0:\mathds{B}^\mathsf{N} \longrightarrow \mathbb{B}^\mathbb{N}$, $(x_1,\dots,x_{\mathbb{N}}) \longmapsto (\overline{x_1},\dots,\overline{x_\mathsf{N}})$ is the Boolean vectorial negation.
We consider the directed graph $\mathcal{G}_g$, where:

\begin{itemize}
\item vertices are all the $\mathbb{N}$-bit words.
\item there is an edge $m \in \llbracket 0, 2^{\mathbb{N}}-1 \rrbracket$ from $x$ to $\check{x}$ if and only if $g(m,x)=\check{x}$.
\end{itemize}

If $\mathbb{G}_g$ is strongly connected, then $G_g$ is strongly transitive.
\end{proposition}

We have then proven that,
\begin{proposition}
\label{prop:regularity}
If $\mathcal{G}_g$ is strongly connected, then $G_g$ is regular.
\end{proposition}

According to Propositions~\ref{prop:transitivity} and~\ref{prop:regularity}, we can conclude that, depending on $g$, if the directed graph $\mathcal{G}_g$ is strongly connected, then the CBC mode of operation is chaotic according to Devaney, as established in our previous research work~\cite{Abdessalem2016}.
In this article and for illustration purpose, we have also given some examples of encryption functions making this mode a chaotic one. 

In the next section we will recall some quantitative measures of chaos that have already been proven in our previous research work.
\subsection{Quantitatives measures}
In~\cite{abidi2016quantitative}, we have respectively developed these two following propositions.
\begin{proposition}
\label{prop:sensitivity}
 The CBC mode of operation is sensible to the initial condition, and its constant of sensibility is larger than the length $\mathsf{N}$ of the block size.
 \end{proposition}

 \begin{proposition}
\label{prop:expansitivity}
The CBC mode of operation is not expansive.
\end{proposition}

To sum up, CBC mode of operation is sensible to the initial conditions but it is not expansive. Let us now investigate new original aspects of chaos of the CBC mode of operation.
 
\section{Topological mixing}
\label{topo}
The topological mixing is a strong version of transitivity.
\begin{definition} 
A discrete dynamical system is said \emph{topologically mixing} if and only if, for any couple of disjoint open set $\mathbb{U},\mathbb{V} \neq \varnothing$, there exists an integer $n_0\in \mathbb{N}$ such that, for all $n > n_0$, $f^{\circ n}(\mathbb{U}) \cap \mathbb{V} \neq
\varnothing$.
\end{definition}
\begin{proposition}
\label{prop:topological mixing}
$(\mathbb{X}, G_g)$ is topologically mixing.
\end{proposition}
This result is an immediate consequence of the lemma below.

\begin{lemma}
For any open ball  $\mathcal{B}=\mathcal{B}((x,m),\varepsilon)$ of $\mathcal{X}$, an index $n$ can be found such that $G_{g}^{\circ n}(\mathcal{B}) = \mathcal{X}$.
\end{lemma}

\begin{proof}
Let $\mathcal{B}$ be an open ball whose radius $\varepsilon$ can be considered as strictly lower than 1, $\varepsilon < 1$.

All the elements of $\mathcal{B}$ have the same state $x$, and due to the definition of the chosen metric, they are such that an integer $k \left(=-\lfloor \log_{10}(\varepsilon)\rfloor\right)$ satisfies:
\begin{itemize}
\item all the strategies of $\mathcal{B}$ have the same $k$ first block messages,
\item after the index $k$, all values are possible.
\end{itemize}

Then, after $k$ iterations, the new state of the system is $G_{g}^{\circ k}(x,m)_1$ and all the strategies are possibles (any point of the form $(G_{g}^{\circ k}(x,m)_1,\textrm{\^{m}})$, with any $\textrm{\^{m}} \in \mathcal{S}$, is reachable from $\mathcal{B}$).

Let $(x',m') \in \mathcal{X}$. We will prove that it can be reached by starting from $\mathcal{B}$.
Indeed, let us consider the point $(\check{x},\check{m})$ of $\mathcal{B}$ defined by:
\begin{itemize}
    \item $\check{x} = x$
    \item $\forall i \leqslant k, \check{m}_i = m_i$,
    \item $\check{m}_{k+1} = G_g^{\circ k}((x,m))_1 \oplus \mathbb{D}_{\kappa}(x')$,
    \item $\forall i \geqslant k+2, \check{m}_i =m_{i-k-2}'$.
\end{itemize}
This latter is such that:\\

\begin{tabular}{ll}
$G_g^{\circ k+1}((\check{x},\check{m}))_1$ & = $G_g(G_g^{\circ k}((x,m))_1, \check{m}_{k+1})_1$\\
& = $G_g(G_g^{\circ k}((x,m))_1, G_g^{\circ k}((x,m))_1 \oplus \mathcal{D}_{\kappa}(x'))_1$\\
& = $\mathcal{E}_{\kappa}\left(G_g^{\circ k}((x,m))_1 \oplus \left( G_g^{\circ k}((x,m))_1 \oplus \mathcal{D}_{\kappa}(x')\right) \right)$\\
& =
$\mathcal{E}_{\kappa}\left(\left(G_g^{\circ k}((x,m))_1 \oplus  G_g^{\circ k}((x,m))_1\right) \oplus \mathcal{D}_{\kappa}(x') \right)$\\ 
& = $\mathcal{E}_{\kappa}\left(\mathcal{D}_{\kappa}(x')\right)$\\
& =$x'$\\
\end{tabular}

\noindent and
$G_g^{\circ k+1}(\check{x},\check{m})_2 = m'$. 

This shows that $(x',m')$ has been reached starting from $\mathcal{B}$. This fact concludes the proof of the lemma and of the proposition claimed previously.
\end{proof}

\section{Topological entropy}
 \label{sec:entro}
Another important tool to measure the chaotic behavior of a dynamical system is the topological entropy, which is defined only for compact topological spaces. Before studying the entropy of CBC mode of operation, we must then check that $(\mathcal{X},\ d)$ is compact.

\subsection{Compactness study}
In this section, we will prove that $(\mathcal{X}, d)$ is a compact topological space, in order to study its topological entropy later. Firstly, as $(\mathcal{X}, d)$ is a metric space, it is separated.
 It is however possible to give a direct proof of this result:

\begin{proposition}
$(\mathcal{X}, d)$ is a separated space.
\end{proposition}

\begin{proof}
Let $(x,w) \neq (\textrm{\^{x}},\textrm{\^{w}})$ two points of $\mathbb{X}$.
\begin{enumerate}
    \item If $x \neq \textrm{\^{x}}$, then the intersection between the two balls  $\mathbb{B}\left((x,w),\frac{1}{2}\right)$ and $\mathbb{B}\left((\textrm{\^{x}},\textrm{\^{w}}), \frac{1}{2}\right)$  is empty.
    \item Else, it exists $k\in\mathds{N}$ such that $w_k \neq \textrm{\^{w}}_k$, then the balls $\mathbb{B}\left((x,w),10^{-(k+1)}\right)$ and $\mathbb{B}\left((\textrm{\^{x}},\textrm{\^{w}}), 10^{-(k+1)}\right)$ can be chosen.
\end{enumerate}

\end{proof}

Let us now prove the compactness of the metric space $(\mathcal{X}, d)$ by using the sequential characterization of compactness.

\begin{proposition}
$(\mathbb{X}, d)$ is a compact space.
\end{proposition}

\begin{proof}
Let $X=\left((x_n, m_n)\right)_{n \in \mathbb{N}}$ be a sequence of $\mathbb{X}$.

There is at least one Boolean vector that appear in infinite number of times in the first components of this sequence, as $\mathbb{B}^\mathbb{N}$ is finite. Let $\tilde{x}$ the lowest of them and
$I$ the (infinite) subsequence of $X$ constituted by all the block messages having their first component equal to $\tilde{x}$.

The first block messages $(w_n)_0$ of the sequences $w_n \in \llbracket 0, 2^\mathsf{N}-1 \rrbracket^\mathds{N}$ (that are the second components of each couple in the infinite sequence $I_0$) all belong in the finite set $\llbracket 0, 2^\mathbb{N}-1 \rrbracket$, and so at least one word of this finite set appears an infinite number of times in $\left((w_n)_0\right)_{n \in \mathds{N}}$. Let $\omega_0 \in \llbracket 0, 2^\mathsf{N}-1 \rrbracket$ be the lowest value occurring an infinite number of times in $I$, and $n_0$ the index of its first occurrence, such that $x_{n_0} = \tilde{x}$, $\left(w_{n_0}\right)_0 = \omega_0$.

Similarly, the subsequence $I_1$ of $X$ constituted by the block messages $(x_n, w_n)$ such that $x_n=\tilde{x}$ and $\left(w_{n}\right)_0 = \omega_0$ is infinite, while all the $\left(w_{n}\right)_1$ belong in $\llbracket 0, 2^\mathbb{N}-1 \rrbracket$. So at least one element of $\llbracket 0, 2^\mathbb{N}-1 \rrbracket$ appears an infinite number of times in the second block messages of the second components $\left(w_{n}\right)_1$ of $I_1$. Let $\omega_1$ be the lowest value in $\llbracket 0, 2^\mathsf{N}-1 \rrbracket$ occurring an infinite number of times at this position, and $n_1$ the index in $X$ of its first occurrence.

We can define again a subsequence $I_2 = (x_n, w_n)$ of $X$ such that $\forall n, x_n = \tilde{x}$, $(w_n)_0=\omega_0$, and $(w_n)_1=\omega_1$, and a similar argument leads to the definition of $\omega_2$, the lowest value in $\llbracket 0, 2^\mathbb{N}-1 \rrbracket$ appearing an infinite number of times in the third block messages of the sequences $w_n \in \llbracket 0, 2^\mathsf{N}-1 \rrbracket^\mathbb{N}$ of $I_3$. This process can be continued infinitely.

Let us finally define the point $l=\left(\tilde{x}, \left(w_{n_k}\right)_k\right)$ of $\mathcal{X}$; the subsequence $\left(x_{n_k}, w_{n_k}\right)$ of $X$ converges to $l$. As for all sequences in $\mathcal{X}$ we can extract a subsequence that converges in $\mathcal{X}$, we can conclude to the compactness of $\mathcal{X}$.
\end{proof}

\subsection{Topological entropy}

Let $(X, d)$ be a compact metric space and $f: X \rightarrow X$ be a continuous map. For each natural number $n$, a new metric $d_n$ is defined on $X$ by

$$d_n(x,y)=\max\{d(f^{\circ i}(x),f^{\circ i}(y)): 0\leq i<n\}.$$

Given any $\varepsilon > 0$ and $n \geqslant 1$, two points of $X$ are $\varepsilon$-close with respect to this new metric if their first $n$ iterates are $\varepsilon$-close (according to $d$).

This metric allows one to distinguish in a neighborhood of an orbit the points that move away from each other during the iteration from the points that travel together. A subset $E$ of $X$ is said to be $(n, \varepsilon)$-separated if each pair of distinct points of $E$ is at least $\varepsilon$ apart in the metric $d_n$.

\begin{definition}
Let $H(n, \varepsilon)$ be the maximum cardinality of a $(n, \varepsilon)$-separated set, the \emph{topological entropy} of the map $f$ is defined by (see \textit{e.g.},~\cite{Adler65} or~\cite{Bowen})
$$h(f)=\lim_{\epsilon\to 0} \left(\limsup_{n\to \infty} \frac{1}{n}\log H(n,\varepsilon)\right). $$
\end{definition}

We have the result,
\begin{theorem}
Entropy of $(\mathcal{X},G_g)$ is infinite.
\end{theorem}

\begin{proof}
Let $\textrm{x}, \textrm{\v{x}}\in \mathbb{B}^\mathbb{N}$ such that $\exists i_0 \in \llbracket 1, N \rrbracket, \textrm{x}_{i_0} \neq \textrm{\v{x}}_{i_0}$. Then, $\forall \textrm{w}, \textrm{\v{w}} \in \mathbb{S}_\mathbb{N}$,
$$d((\textrm{x},\textrm{w});(\textrm{\v{x}},\textrm{\v{w}})) \geqslant 1$$
But the cardinal $c$ of $\mathcal{S}_\mathsf{N}$ is infinite, then $\forall n \in \mathbb{N}, c >e^{n^2}$.

So for all $n \in \mathbb{N}$, the maximal number $H(n,1)$ of $(n,1)-$separated points is greater than or equal to $e^{n^2}$, and then
$$h_{top}(G_g,1) = \overline{lim} \frac{1}{n} log \left( H(n,1)\right) > \overline{lim} \frac{1}{n} log \left( e^{n^2} \right) = \overline{lim} ~(n) = + \infty.$$

\noindent But $h_{top}(G_g,\varepsilon)$ is an increasing function when  $\varepsilon$ is decreasing, then

$$h_{top} \left( G_g \right) = \lim_{\varepsilon \rightarrow 0} h_{top}(G_g,\varepsilon) > h_{top}(G_g,1) = + \infty,$$
\noindent which concludes the evaluation of the topological entropy of $G_g$.
\end{proof}

\section{Conclusion and future work}

In this article, we have deepened the topological study for the CBC mode of operation. Indeed, we have regarded if this tool possesses the property of topological mixing. Additionally, other quantitative evaluations have been performed, and the level of topological entropy has been evaluated too.
All of these properties lead to a complete unpredictable behavior for some CBC modes of operation.

In future work, we will investigate the chaotic behavior of other modes of operation. Additionally, we will see how can we exploit these proofs in order to enrich the field of cryptography. We will more specifically focus on the implementation of such modes in order to prevent them from side channel attacks due to their chaos properties.




\bibliographystyle{unsrt}
\bibliography{biblio}
\end{document}